\documentclass[12pt]{article}

\usepackage{fullpage}
\usepackage{amsfonts,amsmath,epsf,epsfig,bbm}
\usepackage{amsthm}

\usepackage{graphicx}

\title{Greedy Matchings in Bipartite Graphs \\ with Ordered Vertex Sets}

\author{Hans Ulrich Simon}

\newtheorem{theorem}{Theorem}[section]

\newtheorem{lemma}[theorem]{Lemma}
\newtheorem{corollary}[theorem]{Corollary}
\newtheorem{remark}[theorem]{Remark}

\newtheorem{example}[theorem]{Example}

\newcommand{\eset}{\emptyset}

\newcommand{\cS}{{\mathcal S}}

\newcommand{\seq}{\subseteq}
\newcommand{\ra}{\rightarrow}
\newcommand{\la}{\leftarrow}

\newcommand{\path}{\stackrel{*}{\rightarrow}}
\newcommand{\htap}{\stackrel{*}{\leftarrow}}

\begin{document}

\maketitle

\begin{abstract}
We define and study greedy matchings in finite bipartite graphs. 
The latter are of the form $G = (L,R,E)$ where the vertex 
classes $L$ and $R$ are equipped with a linear ordering, respectively. 
We should think of the linear orderings as preference relations. 
It is shown that each vertex-ordered bipartite graph has a unique 
greedy matching. The proof uses (a weak form of) Newman's 
lemma.\footnote{In its original form, this lemma states 
that a noetherian relation is confluent iff it is locally confluent.} 
We say a matching $M$ in $G$ is of \emph{order} $k<\infty$ if, on one 
hand, $M$ matches all vertices in $L$ and, on the other hand, 
$M$ leaves the $|R|-k$ (but not more) vertices of lowest preference 
in $R$ unmatched. If $M$ does not match each vertex in $L$, then 
its order is defined as $\infty$. The goal is now to find 
a matching of minimum order in $G$.  We show that the linear 
ordering imposed on $L$ can be chosen so that the resulting greedy 
matching is of minimum order. A similar result is shown for matchings
with the PBT-property (a property that is relevant in the so-called
Preference-Based Teaching model). Suppose that $G = (C,\cS,E)$ is 
(a vertex-ordered version of) the consistency graph associated
with a concept class $C$ and the corresponding family $\cS$ 
of labeled samples. Suppose furthermore that the ordering of
the samples in $\cS$ is chosen such that smaller samples are 
preferred over larger ones. It is shown in~\cite{FGHHT2024} that 
the largest sample that is assigned by the greedy matching to one 
of the concepts in $C$ has size at most $\log |C|$. We show that 
this upper bound is tight modulo a constant by proving a matching 
lower bound for a specific family of concept classes. The study 
of greedy matchings in vertex-ordered bipartite graphs is motivated 
by problems in learning theory like illustrating or teaching concepts 
by means of labeled examples.

\end{abstract}


\section{Introduction} \label{sec:introduction}

Consider a bipartite graph $G = (L,R,E)$ with vertex sets $L$
(the vertices on the left-hand side) and $R$ (the vertices on the 
right-hand side). Imagine that we have imposed a linear ordering,
called preference relation, on both vertex sets. A \emph{greedy matching}
in $G$ is defined as a maximal matching $M$ that can be produced 
edge-by-edge such that each new edge (to be inserted into $M$) is chosen 
in accordance with the following greedy-selection principle: 
\begin{enumerate}
\item
Pick a most preferred vertex $z$ taken either from the still unmatched 
vertices of $L$ or taken from the still unmatched vertices of $R$
(two options for choosing $z$). 
\item
Match the vertex $z$ with the vertex of highest preference that is
a possible partner of~$z$.\footnote{A ``possible partner'' of $z$ 
is any still unmatched vertex that is adjacent to $z$.}
\end{enumerate}
Suppose that $y_1,\ldots,y_N$ is a list of all elements in $R$
such that $y_k$ is the element with the $k$-highest preference.
The \emph{order} of a matching $M$ is defined as the largest 
number $k$ such that $y_k$ is matched by $M$. A matching in $G$
is called \emph{$L$-saturating} if every $x \in L$ is matched by $M$.
We aim at finding an $L$-saturating matching of smallest order.
This problem arises naturally in learning theory, where we have 
a concept class $C \seq 2^X$ over a domain $X$ at the place of $L$ 
and the family $\cS$ of all realizable samples at the place of $R$. 
The graph $G = (C,\cS,E)$ is then the so-called consistency graph 
which contains the edge $(c,S)$ iff the concept $c$ is consistent 
with the sample $S$. A $C$-saturating matching $M$ 
assigns to each concept $c \in C$ a sample $M(c) \in \cS$ so that $c$ 
is consistent with $M(c)$ and so that different samples are assigned
to different concepts. Intuitively, we can think of $M(c)$ as
a sample which illustrates or, perhaps, even \emph{teaches} $c$
by means of labeled examples.

\begin{example}
Suppose that $M$ has the property that $S = M(c)$ does not only
imply that $c$ is consistent with $S$ but also that each concept
of higher preference than $c$ is not consistent with $S$.
Then $c$ can be derived from $S$ simply by choosing the most preferred
concept that is consistent with $S$. We are then in the so-called
Preference-Based Teaching model (= PBT-model). See~\cite{GRSZ2017}
for a formal definition of this model.
\end{example}

In order to discuss greedy matchings in the consistency graph $G = (C,\cS,E)$, 
one has first to impose linear orderings on $C$ and on $\cS$. 
It looks natural, for instance, to prefer simple concepts over 
less simple ones and to prefer samples of small size over samples of
larger size. However, one could also be interested in an ordering
of the concepts in $C$ that is cleverly chosen so that the resulting 
greedy matching is a $C$-saturating matching of smallest possible order.

\bigskip\noindent
The main results in this paper are as follows:
\begin{enumerate}
\item
Each bipartite graph $G$ with linearly ordered vertex sets has 
a unique greedy matching.\footnote{A related, but weaker, result
is contained in~\cite{FGHHT2024}. See Section~\ref{sec:ars-greedy-matchings} 
for a more detailed discussion.} This result is obtained by considering 
an abstract rewriting system whose objects are the matchings in $G$ 
and whose binary relation between the objects has something to do 
with the above greedy-selection principle. The proof of uniqueness 
of the greedy matching will then be a simple application of 
Newman's lemma~\cite{N1942}. See Section~\ref{sec:ars-greedy-matchings} 
for details.
\item
Suppose that $G$ has an $L$-saturating matching. Then the greedy 
matching in $G$ is an $L$-saturating matching of minimum order provided
that the linear ordering of the vertices in $L$ is chosen appropriately.
A similar result holds for $L$-saturating matchings with the so-called
PBT-property (a property that is relevant in the Preference-Based 
Teaching model). See Section~\ref{sec:minimum-order} for details.
\item
Suppose that $G = (C,\cS,E)$ is the consistency graph associated
with a concept class $C$ and the corresponding family $\cS$ of 
labeled samples. Suppose furthermore that the linear ordering
imposed on $\cS$ is chosen such that smaller samples are preferred
over larger ones. Then the largest sample assigned by the greedy 
matching to one of the concepts in $C$ has size 
at most $\log |C|$.\footnote{As explained in more detail
in Section~\ref{sec:order-bounds}, this result is implicitly contained 
in~\cite{FGHHT2024}.} This upper bound is tight modulo a constant,
which is shown by presenting a matching lower bound for a specific
family of concept classes. See Section~\ref{sec:order-bounds} for 
details.
\end{enumerate}

\begin{description}
\item[Tacit Agreement:] All (directed or undirected) graphs considered
in this paper are assumed to be finite.
\end{description}


\section{Abstract Rewriting Systems and Greedy Matchings}
\label{sec:ars-greedy-matchings}

In Section~\ref{subsec:ars} we recall the definition of abstract 
rewriting systems and Newman's lemma~\cite{N1942},
but we restrict ourselves to the finite case. In the finite case,
Newman's lemma can be stated in a simple form and its proof becomes
quite short and straightforward. In Section~\ref{subsec:greedy-matchings}, 
we consider matchings in a vertex-ordered bipartite graph as the objects 
of an abstract rewriting system. Then we define the notion of greedy 
matchings and, as an application of Newman's lemma, we show that 
every vertex-ordered bipartite graph has a unique greedy matching. 
The latter result is more general than a related result in~\cite{FGHHT2024}.

\subsection{Abstract Rewriting Systems (ARS)} \label{subsec:ars}

A \emph{(finite) abstract rewriting system (ARS)} is formally given 
by a directed graph $D = (V,A)$. We write $a \ra b$ if $(a,b) \in A$. 
We write $a \path b$ and say that $b$ is \emph{reachable from $a$} 
if $D$ contains a path from $a$ to $b$. 
$D$ is said to satisfy the \emph{unique-sink 
condition} if, for every $a \in V$, there exists a sink $b \in V$ 
such that $b$ is the only sink that is reachable from $a$.
$D$ is said to be \emph{confluent} (resp.~\emph{locally confluent}) 
\emph{at vertex $a$} if, for every choice of $b,c \in V$ 
such that $a \path b$ and $a \path c$ (resp.~such that $a \ra b$ 
and $a \ra c)$, there exists a vertex $d \in V$ 
such that $b \path d$ and $c \path d$. $D$ is said to be 
\emph{confluent} (resp.~\emph{locally confluent}) if it is confluent 
(resp.~locally confluent) at every vertex $a \in V$. These notions
are related as follows:

\begin{lemma}[Newman's Lemma specialized to finite ARS]
Suppose that $D = (V,A)$ is an acyclic digraph. Then the following 
assertions are equivalent:
\begin{enumerate}
\item $D$ satisfies the unique-sink condition.
\item $D$ is confluent.
\item $D$ is locally confluent.
\end{enumerate}
\end{lemma}

\begin{proof}
Suppose that $D$ satisfies the unique-sink condition. 
Consider a triple $a,b,c \in V$ such that $a \path b$ and $a \path c$.
There is a unique sink $b'$ (resp.~$c'$) such that every path starting 
at $b$ (resp.~starting at $c$) and ending at a sink actually ends 
in $b'$ (resp.~in $c'$). Moreover, because $D$ is acyclic, such 
paths $b \path b'$and $c \path c'$ must exist. Hence we have the 
situation
\[ {\mathbf b'} \htap b \htap a \path c\path {\mathbf c'} \] 
where symbols representing sinks are written in bold.
Applying now the unique-sink condition with start vertex $a$, 
we may conclude that $b' = c'$. It follows from this discussion 
that $D$ is confluent. \\
If $D$ is confluent, then it is all the more true that $D$ is locally 
confluent. \\
Suppose that $D$ is locally confluent. The \emph{height} of $a \in V$
is defined as the length of the longest path from $a$ to some sink. 
We will prove by induction on $n$ that, for every $n \ge 0$, the 
unique-sink condition holds for every start vertex of height at most $n$.
This is certainly true for $n = 0$ (in which case the start vertex is 
a sink itself). Assume it is true for start vertices of height at most $n$. 
Consider a vertex $a \in V$ of height $n+1$ (if there is any). 
Let $P_1$ and $P_2$ be two
arbitrarily chosen paths which start at $a$ and end at a sink, say $P_1$
ends at the sink $b'$ and $P_2$ ends at the sink $c'$. We have to show
that $b' = c'$. This is immediate by induction if $P_1$ and $P_2$ start
with the same arc. Let us therefore assume that $P_1$ starts with the 
arc $a \ra b$, $P_2$ starts with the arc $a \ra c$ and $b \neq c$.
Since $D$ is locally confluent and acyclic, there must exist a 
sink $a'$ such that $b \path a'$ and $c \path a'$. Hence we have
the situation
\[ 
{\mathbf b'} \htap b \la a \ra c \path {\mathbf c'}\ \mbox{ and }\ 
b \path {\mathbf a'} \htap c 
\]
where, as before, symbols representing sinks are written in bold.
The height of $b$ is at most $n$. Hence the unique-sink condition
holds for start vertex $b$. It follows that $b' = a'$. 
The analogous reasoning for start vertex $c$ reveals that $c' = a'$,
Hence $b' = c'$, which concludes the proof.
\end{proof}

\subsection{Greedy Matchings} \label{subsec:greedy-matchings}

Let $G = (L,R,E)$ be a bipartite graph with vertex sets $L$ and $R$
and with edge set $E$. We write $x \sim y$ if $x$ and $y$ are neighbors
in $G$. Assume that we have imposed a linear order $\succ_L$
on $L$ and also a linear order $\succ_R$ on $R$. We will then refer
to $G = (L,R,E)$, equipped with linear $\succ_L$ and $\succ_R$,
as a \emph{vertex-ordered} bipartite graph. In the sequel, we
will simply use the symbol $\succ$ and it will always be clear from context
which of the two linear orderings it refers to. Intuitively we should
think of $\succ$ as a preference relation: $a \succ b$ means that $a$
is preferred over $b$. Consider now the following ARS $D_G = (V,A)$:
\begin{enumerate}
\item
The vertices in $V$ represent matchings in $G$ (including the empty matching), 
i.e., 
\[ V = \{v_M: M \seq E \mbox{ is a matching in $G$}\} \enspace . \]
\item
For each matching $M$ in $G$, we denote by $L_M \seq L$ the set of vertices
in $L$ which are still unmatched by $M$ and which are furthermore adjacent 
to at least one vertex in $R$ that is also still unmatched by $M$. 
The set $R_M \seq R$ is defined analogously.
\item
We include the arc $(v_M,v_{M'}) \in V \times V$ in $A$ if and only 
if the sets $L_M$ and $R_M$ are non-empty and $M'$ is of the 
form $M' = M \cup \{ \{x^*,y^*\} \}$ such that at least one of 
the following two conditions is satisfied:
\begin{description}
\item[L-Condition:]
$x^*$ is the vertex of highest preference in $L_M$ and $y^*$ is the vertex
of highest preference in $\{y \in R_M: y \sim x^*\}$.
\item[R-Condition:]
$y^*$ is the vertex of highest preference in $R_M$ and $x^*$ is the vertex
of highest preference in $\{x \in L_M: x \sim y^*\}$.
\end{description}
\end{enumerate}
A few (rather obvious) remarks are in place here:
\begin{enumerate}
\item 
If $(v_M,v_{M'}) \in A$, then $|M'| = |M|+1$, which implies that $D_G$ 
is acyclic.
\item 
Either $L_M$ and $R_M$ are both empty or they are both non-empty.
In the former case, $M$ is maximal matching and $v_M$ is a sink.
In the latter case, $M$ is not a maximal matching and $v_M$ is 
not a sink.
\item
Each non-sink $v_M$ in $D_G$ has outdegree $1$ or $2$. More precisely, $v_M$
has outdegree $1$ (resp.~$2$) if the vertex of highest preference in $L_M$ 
is adjacent (resp.~not adjacent) to the vertex of highest preference in $R_M$. 
\end{enumerate}
A matching $M \seq E$ in $G$ is called a \emph{maximal greedy matching} 
or simply a \emph{greedy matching} if  $v_M$ is a sink in $D_G$ 
(i.e, $M$ is maximal matching in $G$) and $v_\eset \path v_M$ (i.e., the vertex 
representing $M$ is reachable from the vertex representing the empty matching). 
Since $D_G$ will in general contain many vertices of outdegree $2$,
it seems, at first glance, that there could be many distinct 
greedy matchings in $G$. However, we will show the following result:

\begin{theorem} \label{th:unique-greedy-matching}
Every vertex-ordered bipartite graph $G = (L,R,E)$  has a unique 
greedy matching.
\end{theorem}

\begin{proof}
It suffices to show that $D_G$ satisfies the unique-sink 
condition.\footnote{The unique sink reachable from $v_\eset$ 
then represents the unique greedy matching.}
Since $D_G$ is acyclic, Newman's lemma applies and it
suffices to show that $D_G$ is locally confluent. Clearly $D_G$ is
locally confluent at any vertex $v_M$ of outdegree $0$ or $1$.
Consider therefore a vertex $v_M$ of outdegree $2$, say with an outgoing
arc to $v_{M_1}$ and another outgoing arc to $v_{M_2}$. We may assume
that $M_1 = M \cup \{ \{x_1^*,y_1^*\} \}$ satisfies the $L$-condition 
while $M_2 = M \cup \{ \{x_2^*,y_2^*\} \}$ satisfies the $R$-condition.
It follows that $x_1^*$ is the vertex of highest preference in $L_M$
and $y_1^*$ is the vertex of highest preference in $\{y \in R_M: y \sim x_1^*\}$.
Similarly, $y_2^*$ is the vertex of highest preference in $R_M$
and $x_2^*$ is the vertex of highest preference in $\{x \in L_M: x \sim y_2^*\}$.
Moreover $x_1^* \neq x_2^*$ and $y_2^* \neq y_1^*$ (because, otherwise,
the outdegree of $v_M$ would be $1$). Consider the matching 
\[ 
M_3 := M \cup \{ \{x_1^*,y_1^*\} , \{x_2^*,y_2^*\} \} 
= M_1 \cup \{ \{x_2^*,y_2^*\} \} = M_2 \cup \{ \{x_1^*,y_1^*\} \}
\enspace.
\]
With respect to $M_1$, $M_3$ satisfies the $R$-condition while,
with respect to $M_2$, $M_3$ satisfies the $L$-condition.
Therefore $v_{M_1} \ra v_{M_3}$ and $v_{M_2} \ra v_{M_3}$.
It follows from this discussion that $D_G$ is locally confluent.
\end{proof}

Let $M_G$ denote the unique greedy matching in $G$. Let $L_M$
and $R_M$ be defined as in the definition of $D_G$ above.
Consider the following greedy procedure $P_{greedy}$ for an 
edge-by-edge generation of~$M_G$.
\begin{description}
\item[Initialization:] $M \la \eset$.
\item[Main Loop:]
If $M$ is a maximal matching in $G$, then return $M_G = M$ and stop.
Otherwise let $x^*$ be the most preferred vertex in $L_M$, 
let $y^*$ be the most preferred vertex in $\{y \in R_M: y \sim x^*\}$,
and insert the edge $\{x^*,y^*\}$ into $M$.
\end{description} 
Let $P'_{greedy}$ be the corresponding procedure with exchanged roles 
of $L$ and $R$. It was shown in~\cite{FGHHT2024} (without using 
ARS-theory) that the procedures $P_{greedy}$ and $P'_{greedy}$ construct 
the same matching. In terms of the paths in $D_G$, this means that
\emph{two very special paths} from $v_\eset$ to a sink --- the path 
all of whose arcs have been inserted into $A$ on base of the $L$-condition
and the path all of whose arcs have been inserted into $A$ on base of
the $R$-condition --- end at the same sink.  
Theorem~\ref{th:unique-greedy-matching} is more general and implies 
that \emph{all} paths from $v_\eset$ to a sink must necessarily end 
at the same sink.

\section{L-Saturating Matchings of Minimum Order}
\label{sec:minimum-order}

Let $G = (L,R,E)$ be a bipartite graph with a linearly ordered vertex 
set $R$. Let $y_1 \succ y_2 \succ\ldots\succ y_N$ be an ordered list 
of the vertices in $R$.  We say that a matching $M \seq E$ in $G$ 
\emph{saturates $L$} if every vertex in $L$ is matched by $M$. 
The \emph{order} of $M$ is defined as the smallest number $k$ 
such that $M$ is an $L$-saturating matching that leaves $y_{k+1},\ldots,y_N$ 
unmatched. If $M$ is not $L$-saturating, then its order is $\infty$.

In order to make the greedy matching in $G$ unique, we will have
to commit ourselves to a linear ordering of the vertices in $L$.
For every fixed ordering $\succ$, we denote the resulting greedy
matching in $G$ by $M^\succ_G$. We will now pursue the question whether 
an $L$-saturating matching of minimum order can be obtained in a greedy 
fashion. We begin with the following result:

\begin{theorem} \label{th:greedy-optimal}
Suppose that $G = (L,R,E)$ is a bipartite graph with a linearly ordered 
vertex set $R$. Then, for every matching $M$ in $G$, there exists a linear 
ordering $\succ$  of the vertices in $L$ such that the order of $M^\succ_G$ 
is less than or equal to the order of $M$ (which implies that, if $M$ 
is $L$-saturating, then $M^\succ_G$ is $L$-saturating too).
\end{theorem}

\begin{proof}
If $M$ is not $L$-saturating, then its order is $\infty$ and there is
nothing to show. We may therefore assume hat $M$ is $L$-saturating.
Let $m = |L|$ and let $k$ be the order of $M$. 
Let $y_{i(1)} \succ\ldots\succ y_{i(m)}$ with $1 \le i(1) <\ldots< i(m) \le k$ 
be an ordered list of those vertices in $\{y_1,\ldots,y_k\}$ that have an
$M$-partner in $L$. For $j=1,\ldots,m$, let $x_j$ denote the vertex in $L$ 
with $M$-partner $y_{i(j)}$. Now impose the linear ordering $x_1 \succ\ldots\succ x_m$ 
on the vertices in $L$. In other words, the vertices in $L$ inherit the 
ordering from their $M$-partners in $R$. Think of $M^\succ_G$ as being 
generated edge-by-edge by the procedure $P_{greedy}$.  Let $y_{i'(j)}$ 
denote the $M^\succ_G$-partner of $x_j$. By a straightforward 
induction, we get $i'(j) \le i(j)$ for $j = 1,\ldots,m$.  Since 
the matching $M$ leaves the vertices $y_{k+1},\ldots,y_N$ unmatched, 
this is all the more true for the matching $M^\succ_G$.
\end{proof}

\noindent
From Theorem~\ref{th:greedy-optimal}, the following result is immediate:

\begin{corollary} \label{cor:greedy-optimal}
Suppose that $G = (L,R,E)$ is a bipartite graph that has an
$L$-saturating matching and its vertex set $R$ is equipped 
with a linear ordering.  Then there exists a linear ordering $\succ$ 
of the vertices in $L$ such that $M^\succ_G$ is an $L$-saturating 
matching of minimum order.
\end{corollary}

Let now $G = (L,R,E)$ be a bipartite graph both of whose vertex sets
are equipped with a linear ordering. An $L$-saturating matching $M$
has the \emph{PBT-property}\footnote{PBT = Preference-Based Teaching} 
if, for every triple $x,y,x'$ such that $y$ is the $M$-partner of $x$ 
and $x' \succ x$, it follows that $x' \not\sim y$. In other words, 
the $M$-partner $x$ of $y$ is the vertex of highest preference 
in $\{x \in L: x \sim y\}$.

The following greedy procedure $P^{PBT}_{greedy}$ is quite similar 
to the procedure $P_{greedy}$ but it is tailor-made so as to return an 
L-saturating matching $M^{PBT}_G$ with the PBT-property (if possible):
\begin{description}
\item[Initialization:] $M \la \eset$.
\item[Main Loop:]
\begin{enumerate}
\item
If $L_M = \eset$ then return $M^{PBT}_G = M$ and stop. Otherwise
let $x^*$ be the vertex of highest preference in $L_M$, set
\[
R'_M \la
\{y \in R_M: y \sim x^* \mbox{ but } y \not\sim x' \mbox{ for every } 
x' \succ x\} \enspace ,
\] 
and proceed with the next instruction.
\item
If $R'_M = \eset$, then return $M^{PBT}_G = M$ and stop.\footnote{In this case,
the returned matching $M^{PBT}_G$ is not $L$-saturating.} Otherwise
let $y^*$ be the most preferred vertex in $R'_M$ and insert 
the edge $\{x^*,y^*\}$ into $M$.
\end{enumerate}
\end{description}

We refer to the matching $M^{PBT}_G$ that is returned by 
procedure $P^{PBT}_{greedy}$ as the {greedy matching
with the PBT-property}. This matching is optimal in the
following sense:

\begin{theorem} \label{th:greedy-pbt}
Suppose that $G = (L,R,E)$ is a vertex-ordered bipartite graph 
that has an $L$-saturating matching with the PBT-property. 
Then $M^{PBT}_G$ is also $L$-saturating. Moreover, among all $L$-saturating 
matchings with the PBT-property, $M^{PBT}_G$ is of the smallest order.
\end{theorem}

\begin{proof}
Let $M$ be an arbitrary $L$-saturating matching with the PBT-property.
Let $x_1 \succ\ldots\succ x_m$ be an ordered list of all elements in $L$. 
Similarly, let $y_1 \succ\ldots\succ y_N$ be an ordered list of all elements
in $R$. Let $k$ be the order of $M$, i.e., $y_k$ is matched by $M$
but $y_{k+1},\ldots,y_N$ are unmatched. Let $y_{i(1)},\ldots,y_{i(m)}$ 
be the $M$-partners of $x_1\ldots,x_m$, respectively. Since $M$ has the 
PBT-property, it follows that $y_{i(j)} \sim x_j$ and $y_{i(j)} \not\sim x_{j'}$
for all $j' = 1,\ldots,j-1$. Thus, when $P^{PBT}_{greedy}$ is about to match $x_j$,
we can be sure that $y_{i(j)}$ is still unmatched. Hence $P^{PBT}_{greedy}$
matches $x_j$ with $y_{i(j)}$ or even with a $y_i$ for some $i < i(j)$. 
We may conclude from this discussion that $M^{PBT}_G$ is an $L$-saturating
matching with the PBT-property that leaves the vertices $y_{k+1},\ldots,y_N$
unmatched. Since $M$ had been chosen arbitrarily, it follows that $M^{PBT}_G$
is of the smallest possible order.
\end{proof}

\section{Bounds on the Order of $L$-Saturating Matchings}
\label{sec:order-bounds}

In this section, we will restrict ourselves to the following setting
(which is the standard setting for binary classification problems in
learning theory):
\begin{itemize}
\item
At the place of the vertex set $L$, we have a class $C$ of concepts 
over some finite domain $X$, i.e., $C \seq 2^X$. Each $c \in C$ is a 
function $c:X \ra \{0,1\}$ which assigns a binary label to each point 
in the domain. 
\item
A pair $(x,b) \in X \times \{0,1\}$ is called a \emph{labeled example}. 
A set $S$ of labeled examples is called a \emph{sample}.
A concept $c \in C$ is consistent with a sample $S$ if, 
for every $(x,b) \in S$, we have that $c(x) = b$. A sample $S$ 
is called \emph{realizable by $C$} if $C$ contains a concept that 
is consistent with $S$.
\item
At the place of $R$, we have the family $\cS$ of all $C$-realizable 
samples. 
\item
The bipartite graph $G$ is chosen as (a vertex-ordered version 
of) the \emph{consistency graph} associated with $C$ and $X$, i.e., 
$G = (C,\cS,E)$ where 
\[
E = \{(c,S) \in C \times \cS: \mbox{$c$ is consistent with $S$}\}
\enspace .
\]
\item
Samples of smaller size are preferred over samples of larger size,
but ordering by cardinality is a partial ordering only. 
We assume that the $\succ$-ordering imposed on $\cS$ is a linear 
extension of the (partial) cardinality-ordering. 
\item
Let $S_1\succ\ldots\succ S_N$ be an ordered list of all $C$-realizable 
samples. Let $M$ be a matching of order $k$ in $G$. Then $S_k$ is the 
largest sample that is matched by $M$ and $|S_k|$ is called the \emph{cost 
caused by $M$}.
\end{itemize}

Since any concept $c \in C$ coincides with the unique concept in $C$ 
that is consistent with the full sample $\{(x,c(x)): x \in X\}$,
it is clear that $G$ has a $C$-saturating matching.
In particular, the greedy matching in $G$ is $C$-saturating.
In~\cite{FGHHT2024}, the authors raise the question of how the cost
caused by the greedy matching in $G$ can be bounded from above in terms 
of $|C|$. They make the following observation:

\begin{remark}[\cite{FGHHT2024}] \label{rem:ub}
Let $G = (C,\cS,E)$ be (a vertex-ordered version of) the consistency graph 
associated with concept class $C$. Let $M_G$ be the greedy matching in $G$ 
and let $q$ denote its cost. Consider a concept $c \in C$ whose $M_G$-partner 
is a sample $S \in \cS$ of size $q$. Because the greedy matching
would prefer, as a partner for $c$, any still unmatched proper subsample 
of $S$ over $S$, it follows that all proper subsamples of $S$ in $\cS$
are already matched with concepts of higher preference than $c$.
Since there are $2^{q-1}$ proper subsamples of $S$ in $\cS$, we may
conclude that $|C| \ge 2^{q-1}+1 > 2^{q-1}$.
\end{remark} 

\noindent
From this observation, the following upper bound is immediate: 

\begin{corollary}
For any finite domain $X$, any concept class $C \seq 2^X$ and any 
linear ordering imposed on $C$, the cost caused by the greedy matching 
in the consistency graph $G = (C,\cS,E)$ is bounded from above 
by $\lceil \log |C| \rceil$.
\end{corollary}

\begin{proof}
Let $q$ denote the cost caused by the greedy matching in $G$.  
The inequality $|C| > 2^{q-1}$ from Remark~\ref{rem:ub}
can be rewritten as $q < 1 + \log |C|$. Since $q$ is an 
integer, it follows that $q \le \lceil \log |C| \rceil$.
\end{proof}

Somewhat surprisingly perhaps, this simply knitted upper bound is tight 
(modulo a small constant). More precisely, we will show the following 
result:
 
\begin{theorem}
There exists a universal constant $\gamma_0$ such that,
for any finite domain $X$ and for any linear ordering
imposed on $C_{all} = 2^X = \{c: X \ra \{0,1\}\}$, 
the cost caused by the greedy matching in the consistency graph $G = (C_{all},\cS,E)$ 
is bounded from below by $\gamma_0 \cdot \log |C_{all}|$.
\end{theorem}

\begin{proof}
For sake of brevity, set $n = |X|$ so that $|C_{all}| = 2^n$ 
and $\log|C_{all}| = n$. For $d=0,1,\ldots,n$, 
set $\Phi_d(n) = \sum_{i=0}^{d} \binom{n}{i}$.  It is well known~\cite{CRS1994} 
that $\Phi_d(n) \le \left(\frac{en}{d}\right)^d$. The number of labeled samples 
of size at most $d$ equals 
\[
\sum_{i=0}^{d} \binom{n}{i} 2^i \le 2^d \Phi_d(n) \le \left(\frac{2en}{d}\right)^d
\enspace .
\]
Let $d^*$ be the smallest number $d \in \{1,\ldots,n\}$
such that $\left(\frac{2en}{d}\right)^d \ge 2^n$. For obvious reasons, 
the order of the greedy matching is bounded from below by $d^*$. 
Setting $\gamma^* = d^*/n \le 1$, we obtain the 
inequality $\left(\frac{2e}{\gamma^*}\right)^{\gamma^* n} \ge 2^n$,
or equivalently, $\left(\frac{2e}{\gamma^*}\right)^{\gamma^*} \ge 2$.
Consider the function $h(\gamma) = \left(\frac{2e}{\gamma}\right)^\gamma$. 
Note that $h(1) = 2e > 2$. 
\begin{description}
\item[Claim:]
The function $h(\gamma)$ is increasing in the interval $(0,1]$
and $\lim_{\gamma \ra 0}h(\gamma) = 1$.
\end{description}
The claim implies that there is a unique real number $0 < \gamma_0 \le 1$ 
which satisfies $h(\gamma_0) = 2$ and 
$\gamma_0 = \min\{\gamma \in (0,1]: h(\gamma) \ge 2\}$.\footnote{Numerical 
calculations reveal that $0.214 < \gamma_0 < 0.215$.} 
Since $h(\gamma^*) = \left(\frac{2e}{\gamma^*}\right)^{\gamma^*} \ge 2$
and $0 < \gamma^* \le 1$, it follows that $\gamma^* \ge \gamma_0$ 
and $d^* \ge \gamma_0 n$. \\
Although the verification of the above claim is an easy application 
of calculus, we will now present a proof thereby making the paper
more self-contained. The function $h$ can written in the form
\[ 
h(\gamma) = \exp\left[\gamma \cdot \ln\left(\frac{2e}{\gamma}\right)\right] =
\exp(g(\gamma))\ \mbox{ for }\   
g(\gamma) = \gamma \cdot \ln\left(\frac{2e}{\gamma}\right) \enspace .
\]
Since the exponential function is continuous and monotonically increasing,
it suffices to show that $g(\gamma)$ is increasing in 
the interval $(0,1]$ and $\lim_{\gamma \ra 0}g(\gamma) = 0$. 
Building the first derivative of $g$, we get 
\[
g'(\gamma) = \ln\left(\frac{2e}{\gamma}\right) 
+ \gamma \cdot \frac{\gamma}{2e} \cdot \frac{-2e}{\gamma^2} =
\ln\left(\frac{2e}{\gamma}\right) - 1 \enspace .
\]
Clearly $g'(\gamma) > 0$ for all $0 < \gamma < 2$. Hence $h(\gamma)$
is increasing in the interval $(0,1]$ (and even in the
interval $(0,2)$). We observe that
\[
\lim_{\gamma \ra 0}g(\gamma) = 
\lim_{\gamma \ra 0}\left( \gamma \cdot \ln(2e)\right) + 
\lim_{\gamma\ra 0}\left(\gamma \cdot \ln\left(\frac{1}{\gamma}\right) \right) = 
0 + \lim_{x\ra\infty}\frac{\ln(x)}{x} = 0 \enspace ,
\]
which completes the proof of the claim.
\end{proof}

\bibliography{rep}

\end{document}